\documentclass[]{article}
\pdfoutput=1

\usepackage[none]{hyphenat}
\usepackage[margin=1.5in]{geometry}

\usepackage{natbib}

\usepackage{float}
\usepackage{graphicx}
\usepackage{amsmath}
\usepackage{amsthm}
\usepackage{amsfonts}
\usepackage{amssymb}
\usepackage[colorlinks=true,linkcolor=red,citecolor=red,filecolor=red,urlcolor=red]{hyperref}
\usepackage{doi}

\newtheorem{thm}{Theorem}

\newtheorem{cor}[thm]{Corollary}
\newtheorem{lem}[thm]{Lemma}
\newtheorem{defn}[thm]{Definition}

\newcommand{\for}{\mbox{ for }}
\newcommand{\minpoly}{\mbox{minpoly}}
\newcommand{\dirsum}{\bigoplus}

\newcommand{\dsum}{\displaystyle\sum\limits}

\newcommand{\pmp}[1][q,b]{\boldmath{P}_{#1}} 
\newcommand{\pmpmin}[1][q,b]{\boldmath{P}_{#1}} 
\newcommand{\prank}[1][t]{\boldmath{Q}_{#1}} 

\def\algend{\ensuremath{\Box}}
\def\Prob{\mbox{Prob}}

\def\F{{\mathbb F}}
\def\K{{\mathbb K}}
\def\rank{{\mbox{rank}}}
\def\mod{{\mbox{mod}}}
\def\lcm{{\mbox{lcm}}}
\def\H{{\tt H}}

\title{
Probabilistic analysis of Wiedemann's algorithm for minimal polynomial
computation
\thanks{Research supported by National Science Foundation Grants CCF-1018063 and CCF-1016728}
}

\author{Gavin Harrison\\
Drexel University \and
Jeremy Johnson\\
Drexel University \and
B. David Saunders\\
University of Delaware, 
}

\date{{\small \doi{10.1016/j.jsc.2015.06.005}\\
\copyright\ 2015, Elsevier. Licensed under the Creative Commons Attribution-NonCommercial-NoDerivatives 4.0 International\\ \url{http://creativecommons.org/licenses/by-nc-nd/4.0/}}}

\begin{document}

\maketitle

\begin{abstract}
Blackbox algorithms for linear algebra problems start with 
projection of the sequence of powers of a matrix to a sequence of vectors (Lanczos), a sequence of scalars (Wiedemann) or a sequence of smaller matrices (block methods).  
Such algorithms usually depend on the minimal polynomial of the resulting sequence being that of the given matrix.  
Here
exact formulas are given 
for the probability that this occurs.  
They are based on the generalized Jordan normal form (direct sum of companion matrices of the elementary divisors) of the matrix.
Sharp bounds follow from this for matrices of unknown elementary divisors. 
The bounds are valid for all finite field sizes and show that a small blocking factor can give high probability of success for all cardinalities and matrix dimensions.
\end{abstract}


\section{Introduction}
\label{introduction}
The minimal polynomial of a $n\times n$ matrix $A$ may be viewed as the minimal scalar generating polynomial of the linearly recurrent sequence of powers of $\bar{A} = (A^0, A^1, A^2, A^3, \ldots)$.
Wiedemann's algorithm \citep{Wiedemann86} projects the matrix sequence 
to a scalar sequence $s = (s_0, s_1, s_2, \ldots$), where $s_i = u^TA^iv$. 
The vectors $u,v$ are chosen at random.  The algorithm continues by computing the minimal generating polynomial of $s$ which, with high probability, is the minimal polynomial of $A$.
Block Wiedemann algorithms \citep{Coppersmith94,EGGSV06,Kaltofen95,Vill97,Vill99} fatten $u^T$ to matrix $U$ having several rows and $v$ to a matrix $V$ having multiple columns, so that 
the projection is to a sequence of smaller matrices, $B = U\bar{A}V = (UA^0V, UA^1V, UA^2V, \ldots)$, 
where, for chosen block size $b$, 
$U,V$ are uniformly random matrices of shape $b\times n$ and $n\times b$, respectively.
A block Berlekamp/Massey algorithm is then used to compute the matrix minimal generating polynomial of B \citep{Yuhasz13,GJV03}, and from it the minimal scalar generating polynomial.
All of the algorithms based on these random projections rely on preservation 
of some properties, including at least the minimal generating polynomial. 
In this paper we analyze the probability of preservation of minimum polynomial under random projections for a matrix over a finite field.

Let $A \in \F_q^{n \times n}$ 
and let {\bf $\pmp(A)$} denote the 
probability that $\minpoly(A) = \minpoly(U\bar AV)$ for uniformly random $U \in \F_q^{b\times n}$ and $V \in \F_q^{n\times b}$.
$\pmp(A)$ is the focus of this paper and this notation will be used throughout.
Our analysis proceeds by first giving exact formulas for $\pmp(A)$ in terms of field cardinality $q$, projected dimension $b$, and the elementary divisors of $A$.  
Let $\pmpmin(n) = \min(\{ \pmp(A)~|~{A \in \F_q^{n \times n}}\})$,
a function of field cardinality, $q$, projected block size, $b$, and the matrix dimension, $n$.
Building from our formula for $\pmp(A)$, we 
give a means to compute $\pmpmin(n)$ precisely and hence to derive a sharp lower bound.
Our bound is less pessimistic than earlier ones such as \citep{Kaltofen:1991:SSLS,Kaltofen95} which primarily apply when the field is large.

Even for cardinality 2, we show that a modest block size (such as b = 22) 
assures high probability of preserving the minimal polynomial. 
A key observation is that when the cardinality is small the number 
of low degree irreducible polynomials is also small.  \citet{Wiedemann86} used this observation to make a bound for probability of minimal polynomial preservation in the non-blocked algorithm.  Here, we have exact formulas for $\pmp(A)$ which are worst when the irreducibles in the elementary divisors of $A$ are as small as possible. Combining that with information on the number of low degree irreducibles, we obtain
a sharp lower bound for the probability of minimal polynomial preservation for arbitrary $n\times n$ matrix (when the elementary divisors are not known a priori).

Every square matrix, $A$, over a finite field $\F$ is similar over $\F$ to its generalized Jordan normal form, $J(A)$, a block diagonal direct sum of the Jordan blocks of its elementary divisors, which are powers of irreducible polynomials in $\F[x]$.  $A$ and $J(A)$ have the same distribution of random projections.  Thus we may focus attention on matrices in Jordan form.
After section \ref{defs} on basic definitions concerning matrix structure and linear recurrent sequences, 
the central result, theorem \ref{thm:prob1} is the culmination of section \ref{sec:prob} where probability of preserving the minimal polynomial for a matrix of given elementary divisors is analyzed.  Examples immediately following theorem \ref{thm:prob1} illustrate the key issues.  The exact formulation of the probability of minimal polynomial preservation in terms of matrix, field, and block sizes is our main result, theorem \ref{thm:worst}, in section \ref{sec:worst}. 
It's corollaries provide some simplified bounds. Section \ref{examplebounds}, specifically figure \ref{probability_plot}, illustrates practical applicability.
We finish with concluding remarks, section \ref{conclusion}.

\section{Definitions and Jordan blocks}\label{defs}
Let 
$\F^{m\times n}$ be the vector space of $m\times n$ matrices over $\F$, and $\F_\infty^{m\times n}$ the vector space of sequences of ${m\times n}$ matrices over $\F$. 
For a sequence $S = (S_0,S_1,S_2,...) \in \F_\infty^{m\times n}$ and polynomial $f(x) = \sum_{i=0}^d f_ix^i \in \F[x]$, define $f(S)$ as the sequence whose $k$-th term is $\sum_{i=o}^d f_i S_{i+k}$.  This action is a multiplicative group action of $\F[x]$ on $\F_\infty^{m\times n}$, because $(fg)(S) = f(g(S))$ for $f,g \in \F[x]$ and $f(S+\alpha T) = f(S) + \alpha f(T)$ for $S, T \in \F_\infty^{m\times n}$ and $\alpha \in \F$.
Further, if 
$f(S) = 0$ we say $f$ {\em annihilates} $S$. 
In this case, $S$ is completely determined by $f$ and its leading $d$ coefficient matrices $S_0, S_1, \ldots, S_{d-1}$.
Then $S$ is said to be {\em linearly generated}, and $f(x)$ is also called a {\em generator} of $S$.
Moreover, for given $S$, the set of polynomials that generate $S$ is an ideal of $\F[x]$.  
Its unique monic generator is called the {\em minimal generating polynomial}, or just {\em minimal polynomial} of $S$ and is denoted $\minpoly(S)$. 
In particular, the ideal of the whole of $\F[x]$ is generated by 1 and, acting on sequences, generates only the zero sequence.
For a square matrix $A$,
the minimal polynomial of the sequence 
$\bar{A} = (I, A, A^2, \ldots)$ is also called the minimal polynomial of $A$.
$(\minpoly(A) = \minpoly(\bar A))$.

We will consider the natural transforms of sequences by matrix multiplication on either side.
For $U \in \F^{b\times m},$ 
$US = (US_0, US_1, US_2, \ldots)$ over $\F^{b\times n},$
and for $V \in \F^{n\times b},$ 
$SV = (S_0V, S_1V, S_2V, \ldots)$ over $\F^{m\times b}.$
For any polynomial $g$, it follows from the definitions that $g(USV) = Ug(S)V$.
It is easy to see that the generators of $S$ also generate $US$ and $SV$, so that
$\minpoly(US)~|~\minpoly(S),$ and
$\minpoly(USV)~|~\minpoly(SV)~|~\minpoly(S).$

More specifically, we are concerned with random projections,
$U\bar{A}V$, 
 of a square matrix $A$, 
where $U,V$ are uniformly random, $U \in \F^{b\times n},V \in \F^{n\times b}$.  
By {\em uniformly} random, we mean that each of the (finitely many) matrices of the given shape is equally likely.
\begin{lem}
\label{lem:sim}
Let $A,B$ be similar square matrices over $\F_q$ and let $b$ be any block size.
Then $\pmp(A) = \pmp(B)$.  In particular, $\pmp(A) = \pmp(J)$ where $J$ is the generalized Jordan form of $A$.
\end{lem}
\begin{proof}
Suppose $A$ and $B$ are similar, so that $B = WAW^{-1}$, 
for a nonsingular matrix $W$. The $(U,V)$ projection of $WAW^{-1}$ is the $(UW,W^{-1}V)$ projection of $A$. But when $U,V$ are uniformly random variables, then so are $UW$ and $W^{-1}V$, since the multiplications by $W$ and $W^{-1}$ are bijections.
\end{proof}

Thus, without loss of generality, in the rest of the paper we will restrict attention to matrices in generalized Jordan normal form.  We describe our notation for Jordan forms next.

The {\em companion matrix} of a monic polynomial $f(x) = f_0 + f_1x + \ldots + f_{d-1}x^{d-1} + x^d$ of degree $d$ is 
\[ C_f = 
\left( \begin{matrix}
		0 & 0 & 0 & \cdots & 0 & -f_0 \\
		1 & 0 & 0 & \cdots & 0 & -f_1 \\
		0 & 1 & 0 & \cdots & 0 & -f_2 \\
		\vdots & \vdots & \vdots & \ddots & \vdots & \vdots \\
		0 & 0 & 0 & \cdots & 1 & -f_{d-1}
	\end{matrix}
\right)
\mbox{ and }
J_{f^e} = 
\left( \begin{matrix}
		C_f & 0 & 0 & \cdots & 0 & 0 \\
		I & C_f & 0 & \cdots & 0 & 0 \\
		0 & I & C_f & \cdots & 0 & 0 \\
		\vdots & \vdots & \vdots & \ddots & \vdots & \vdots \\
		0 & 0 & 0 & \cdots & I & C_f
	\end{matrix}
\right)
\]
is the {\em Jordan block} corresponding to $f^e$, a $de\times de$ matrix. 
It is standard knowledge that the minimal polynomial of $J_{f^e}$ is $f^e$.
When $e = 1$, $J_f = C_f$.

In particular, we use these basic linear algebra facts: For irreducible $f$, (1) $f^{e-1}(J_{f^e})$ is zero everywhere except in the lowest leftmost block where it is a nonsingular polynomial in $C_f$ (see, for example, \cite{Robinson70}), and (2) the Krylov matrix $K_{C_f}(v) = (v, C_fv, C_f^2v, \ldots, C_f^{d-1}v)$ is nonsingular unless $v = 0$.

Generalized Jordan normal forms are (block diagonal) direct sums of primary components, 
$$J = \dirsum_i \dirsum_j J_{f_i^{e_{i,j}}},$$ where the $f_i$ are distinct irreducibles and the $e_{i,j}$ are positive exponents, nonincreasing with respect to $j$.
Every matrix is similar to a generalized Jordan normal form, unique up to order of blocks.

\section{Probability Computation, Matrix of Given\\ Structure}
\label{sec:prob}

Recall our definition that, for $A \in \F_q^{n \times n}$, 
{\bf $\pmp(A)$} denotes the 
probability that minimal polynomial is preserved under projection to $b\times b$, i.e., $\minpoly(A) = \minpoly(U\bar AV)$ for uniformly random $U \in \F_q^{b\times n}$ and $V \in \F_q^{n\times b}$.
For the results of this paper the characteristic of the field is not important.  
However the cardinality $q$ is a key parameter in the results.
For simplicity, we are restricting to projection to square blocks. 
It is straightforward to adjust these formulas to the case of rectangular blocking.

By lemma \ref{lem:sim}, we may assume that the given matrix is in generalized Jordan form, which is a block diagonal matrix.
The projections of a block diagonal matrix are sums of independent projections 
of the blocks.  In other words, for the $U,V$ projection of $A = \dirsum A_i$ let $U_i, V_i$ be the blocks of columns of $U$ and rows of $V$ conformal with the block sizes of the $A_i$.  Then $U\bar{A}V = \sum U_i\bar{A_i}V_i$.  In additionto this observation the particular structure of the Jordan form is utilized.

In subsection \ref{toprimarycomp}  
we show that the probability $\pmp(A)$
may be expressed in terms of $\pmp(J(f))$ for 
the primary components, $J(f) = \dirsum_j J_{f^{e_j}}$, 
associated with the distinct irreducible factors of the minimal polynomial of $A$.
This is further reduced to the probability for a direct sum of companion matrices $C_f$ in \ref{todirectsum}. 
Finally, the probability for $\dirsum C_f$ is calculated in \ref{adirectsum} by reducing it 
to the probability that a sum of rank 1 matrices over the extension field $\F_q[x]/\langle f(x) \rangle$ is zero.  
In consequence we obtain a formula for $\pmp(A)$ in theorem \ref{thm:prob1}.
Examples Examples illustrating theorem \ref{thm:prob1} are given in 
subsection \ref{sec:examples}.

\subsection{Reduction to Primary Components}\label{toprimarycomp}

Let $A = \dirsum_i \dirsum_j J_{f_i^{e_{i,j}}} \in \F_q^{n \times n},$ where the $f_i \in \F_q[x]$ are distinct irreducible polynomials and the $e_{i,j}$ are positive exponents, nonincreasing with respect to $j$.  In this section, we show that $$\pmp(A) = \prod_i \pmp\left(\dirsum_j J_{f_i^{e_{i,j}}}\right).$$

\begin{lem}
\label{lem:minpolysum}
Let $S$ and $T$ be linearly generated matrix sequences.
Then $\minpoly(S + T) ~|~ \lcm(\minpoly(S),\minpoly(T))$.  
\end{lem}
\begin{proof}
Let $f = \minpoly(S)$, $g = \minpoly(T)$ and $d = \gcd(f,g)$.  
The lemma follows from the observation that 
\[
(fg/d)(S+T)  =  (fg/d)(S) + (fg/d)(T) 
                    =  (g/d)(f(S)) + (f/d)(g(T)) = 0.
\]
\end{proof}

As an immediate corollary we get equality when $f$ and $g$ are relatively prime.

\begin{cor}
\label{cor:minpolysumrelprime}
Let $S$ and $T$ be linearly generated matrix sequences with $f = \minpoly(S)$ and $g = \minpoly(T)$ such that $\gcd(f,g) = 1$.  Then $\minpoly(S+T) = fg$.
\end{cor}
\begin{proof}
By the previous lemma, $\minpoly(S+T) = f_1g_1$ 
with $f_1~|~f$ and $g_1~|~g$.  We show that $f_1 = f$ and $g_1 = g$.  Under
our assumptions, 
$ 0 = fg_1(S+T)  =  fg_1(S) + fg_1(T) = fg_1(T)$ so that $fg_1$ is a generator of $T$.  But if $g_1$ is
a proper divisor of $g$, then $fg_1$ is not in the ideal generated by $g$, a contradiction.
Similarly $f_1$ must equal $f$.
\end{proof}

\begin{thm}
\label{thm:minpolysum}
Let $A = \dirsum_i \dirsum_j J_{f_i^{e_{i,j}}} \in \F_q^{n \times n},$ where the $f_i$ are distinct irreducibles and the $e_{i,j}$ are positive exponents, nonincreasing with respect to $j$.  Then, $\pmp(A) =\prod_i \pmp\left(\dirsum_j J_{f_i^{e_{i,j}}}\right)$.
\end{thm}
\begin{proof}
Let $S = U\bar{A}V$, and $S_i = U_i \dirsum_j J_{f_i^{e_{i,j}}} V_i$, where $U_i,V_i$ are blocks of $U,V$ conforming to the dimensions of the blocks of $A$.  Then, $S = \sum_i S_i$.  Let $g_i = \minpoly\left(S_i\right)$.  Because $g_i ~|~ f_i^{e_{i,1}}$ and all $f_i$ are unique irreducibles, then $gcd(g_i,g_j) = 1$ when $i \neq j$.  Therefore, by corollary \ref{cor:minpolysumrelprime}, $\minpoly(S) = \prod_i g_i$.  Therefore $\minpoly(S) = \minpoly(A)$ if and only if $\minpoly(S_i) = f_i^{e_{i,1}}$ for all $i$, and $\pmp(A) = \prod_i \pmp\left( \dirsum_j J_{f_i^{e_{i,j}}} \right)$.
\end{proof}

\subsection{Probability for a Primary Component}\label{aprimecomponent}

Next we calculate $\pmp(\dirsum{J_{f^{e_i}}})$, where $f \in \F_q[x]$ is an irreducible polynomial and $e_i$ are positive integers.  We begin with the case of a single Jordan block before moving on to the case of a direct sum of several blocks.

Consider the Jordan block $J \in \F_q^{n \times n}$ determined by an irreducible power, $f^e$.  $\pmp(J)$ is independent of $e$.  Thus, $\pmp(J_{f^e}) = \pmp(C_f)$. This fact and $\pmp(C_f)$ are the subject of the next lemma.

\begin{thm}
\label{thm:block}
Given a finite field $\F_q$, an irreducible polynomial $f(x) \in \F_q[x]$ of degree $d$,
an exponent $e$, and a block size $b$, let $J = J_{f^e} \in \F_q^{de\times de}$ 
be the Jordan block of $f^e$ and 
let $\bar{J}$ be the sequence  $(I, J, J^2, \ldots)$. 
For $U \in \F_q^{b\times de}$ and $V \in \F_q^{de\times b}$ 
	the following properties of minimal polynomials hold.
\begin{enumerate}
\item
If the entries of $V$ are uniformly random in $\F_q$, then
$$\Prob(f^e = \minpoly(\bar J V)) = 1 - 1/q^{db}.$$
Note that the probability is independent of $e$.
\item
If $V$ is fixed and the entries of $U$ are uniformly random in $\F_q$, then
\[
\Prob(\minpoly(\bar J V) = \minpoly( U \bar J V))  \geq 1 - 1/q^{db},
\]
with equality if $V \neq 0$.
\item
If $U$ and $V$ are both uniformly random, then
$$\pmp(J) = \Prob(f^e = \minpoly(U \bar J V)) 
= (1 - 1/q^{db})^2 = \pmp(C_f).
$$
\end{enumerate}
\end{thm}

\begin{proof}
For parts 1 and 2, 
let $M$ be the lower left $d \times d$ block of $f^{e-1}(J)$.  $M$ is nonzero and all other parts of $f^{e-1}(J)$ are zero.  
Note that $\F_q[C_f]$, the set of polynomials in the companion matrix $C_f$, is isomorphic to $\F_q[x]/\langle f \rangle$.
Since $M$ is nonzero and a polynomial in $C_f$, it is nonsingular.  
Since for any polynomial $g$ and matrix $A$ one has $g(\bar A) = \bar Ag(A)$, 
the lower left blocks of the sequence $f^{e-1}(\bar J)$ form the sequence $(M, C_fM, C_f^2M, \ldots) = \bar C_fM$.  

Part 1.  $f^{e-1}(\bar J) V$ is zero except in its lower $d$ rows which are $\bar {C_f} MV_1$,
where $V_1$ is the top $d$ rows of $V$.
This sequence is nonzero with minimal polynomial $f$ unless $V_1 = 0$ which has probability $1/q^{db}$.

Part 2.
If $V = 0$ the inequality is trivially true.  For $V \neq 0$, 
$Uf^{e-1}(\bar J) V$ is zero except in its lower left $d\times d$ corner $U_e \bar{C_f}M V_1$,  
where $V_1$ is the top $d$ rows of $V$ and $U_e$ is the rightmost $d$ columns of $U$.
Since $M$ is nonsingular, $M V_1$ is uniformly random and the question is reduced to the case of projecting a companion matrix.

Let $C = C_f$ for irreducible $f$ of degree $d$.
For nonzero $V \in \F^{d\times b}$, $\bar{C}V$ is nonzero and has minpoly $f$.  We must show that if $U \in \F^{b\times d}$ is nonzero then $U\bar{C}V$ also has minpoly $f$. 
Let $v$ be a nonzero column of $V$.
The Krylov matrix 
$K_C(v) = (v, Cv, C^2v, \ldots, C^{d-1}v)$
has as it's columns the first $d$ vectors of the sequence $\bar{C}v$.
Since $v$ is nonzero, this Krylov matrix is nonsingular and 
$u K_C(v) = 0$ implies $u = 0$.
Thus, for any nonzero vector $u$, we have $u\bar {C} v \neq 0$ so that, for nonzero $U$, the sequence $U\bar C_f V$ is nonzero and has minimal polynomial $f$ as needed.
Of the $q^{db}$ possible $U$, 
only $U = 0$ fails to preserve the minimal polynomial. 

Part 3.
By parts 1 and 2,
we have $(1-1/q^{db})$ probability of preservation of minimum polynomial $f^e$, first at right reduction by $V$ to the sequence $\bar{J}V$ and then again the same probability at the reduction by $U$ to block sequence $U\bar{J}V$.  Therefore, $\pmp(J) = (1-1/q^{db})^2$.
\end{proof}

\subsubsection{Reduction to a Direct Sum of Companion Matrices}\label{todirectsum}

Consider the primary component $J = \dirsum{J_{f^{e_i}}}$,  
for irreducible $f$, and let $e = \max(e_i)$.  
We reduce the question of projections preserving minimal polynomial for $J$ to the corresponding question for direct sums of the companion matrix $C_f$, which is then addressed in the next section.

\begin{lem}
\label{lem:companion}
Let $J = \dirsum J_{f^{e_i}}$, where $f \in \F_q[x]$ is irreducible, and $e_i$ are positive integers.   Let $e = \max(e_i)$.  Let $s$ be the number of $e_i$ equal to $e$.  Then, $$\pmp(J) = \pmp\left(\dirsum_{i=1}^s C_f \right).$$
\end{lem}
\begin{proof}
The minimal polynomial of $J$ is $f^e$ and that of $f^{e-1}(J)$ is $f$. A projection $U{\bar{J}}V$ preserves minimal polynomial $f^e$ if and only if $f^{e-1}(U{\bar{J}}V)$ has minimal polynomial $f$.
For all $e_i < e$ we have $f^{e-1}(J_{f^{e_i}}) = 0$, so it suffices to consider direct sums of Jordan blocks for a single (highest) power $f^e$.

Let $J_e = J_{f^e}$ be the Jordan block for $f^e$,
and let 
$A = \dirsum_{i=1}^s J_e$. 
A projection $U{\bar A}V$ is successful if it has the same minimal polynomial as $A$. 
This is the same as saying the minimal polynomial of $f^{e-1}(U{\bar A}V)$ is $f$.
We have 
\[
f^{e-1}(U{\bar A}V) = U f^{e-1}({\bar A}) V = \sum_{i=1}^s U_i f^{e-1}(\bar{J_e}) V_i = \sum_{i=1}^s U_{i,e} {\bar{C_f}} {\tilde V_{i,1}}.
\]  
For the last expression $U_{i,e}$ is the rightmost block of $U_i$ and ${\tilde V_{i,1}}$ is the top block of $MV_i$.  The equality follows from the observation in the proof of theorem \ref{thm:block} that $f^{e-1}(\bar J)$ is the sequence that has $\bar{C_f}M$ ($M$ nonsingular) in the lower left block and zero elsewhere.
Thus, $\pmp(J) = \pmp\left(\dirsum_{i=1}^s C_f \right)$.
\end{proof}

\subsubsection{Probability for a Direct Sum of Companion Matrices}\label{adirectsum}

Let $f$ be irreducible of degree $d$.
To determine the probability that a block projection of 
$A = \dirsum_{i=1}^t C_f$ 
preserves the minimal polynomial of $A$, we need to determine the probability 
that $\dsum_{i=1}^t U_i \bar{C_f} V_i = 0$.  We show that this is equivalent to the probability that a sum of rank one matrices over 
$\K = \F_q[x]/\langle f(x) \rangle$ is zero and we establish
a recurrence relation for this probability
in corollary \ref{recurrence}.
This may be considered the heart of the paper.

\begin{lem}
\label{thm:zero_sum_prob}
Let $A = \dirsum_{i=1}^t C_f \in \F_q^{n \times n}$, where $f \in \F_q[x]$ is irreducible of degree $d$.  $\pmp(A)$ is equal to the probability that $S = U \bar{A}V  = \dsum_{i=1}^t U_i \bar{C_f} V_i \neq 0$, where $U \in \F_q^{b \times n}$ and $V \in \F_q^{n \times b}$ are chosen uniformly randomly, and $U_i,V_i$ are blocks of $U,V$, respectively, conforming to the dimensions of the blocks of $A$.
\end{lem}
\begin{proof}
Because $\minpoly(S) ~|~ \minpoly(A)$ and $\minpoly(A) = f$, then $\minpoly(S) ~|~ f$.  Because $f$ is irreducible, it has just two divisors: $f$ and $1$.  The divisor 1 generates only the zero sequence.  Therefore, if $S = 0$ then $\minpoly(S) = 1$.  Otherwise, $\minpoly(S) = f$.  Thus $\pmp(A)$ equals the probability that $S \neq 0$.
\end{proof}

The connection between sums of sequences $U \bar{C_f} V$ and sums of rank 
one matrices over the extension field $\K$ is obtained through the 
observation that for column vectors $u,v$, one has $u^T \bar{C_f} v = u^T \rho(v)$ where $\rho$ is the 
regular matrix representation of $\K$, i.e.
$\rho(v)u = v u$ in $\K$.
The vectors $u$ and $v$ can be interpreted as elements of $\K$ by associating
them with the polynomials $u(x) = \sum_{i=0}^{d-1} u_i x^i$ and 
$v(x) = \sum_{i=0}^{d-1} v_i x^i$.  Moreover,
if $\{1,x,x^2,\ldots,x^{d-1}\}$ is chosen as a basis for $\K$ over $\F$, 
then $\rho(x) = C_f$ and 
$\rho(v) = \sum_{i=0}^{d-1} v_i \rho(x)^i = \sum_{i=0}^{d-1} v_i C_f^i$.

Letting $C = C_f$, the initial segment of 
$u^T \bar{C_f} v$ is $u^T ( v, C v, C^2 v,\ldots, C^{d-1}v )$,
which is $u^T K_C(v)$, where $K_C(v)$ is the Krylov matrix whose columns are $C^iv$.  
The following lemma shows that $K_C(v) = \rho(v)$ and establishes the
connection $u^T \bar{C_f} v = u^T \rho(v)$.

\begin{lem}
\label{thm:krillov_reg}
Let $f$ be an irreducible polynomial and $\K = \F[x]/\langle f \rangle$ be the extension field defined by $f$.  Let $\rho$ be the regular representation of
$\K$ and $C = C_f$ the companion matrix of $f$. Then
$\rho(v) = \sum_{j=0}^{d-1} v_j C^j = K_C(v)$.
\end{lem}
\begin{proof}
Let $e_j$ be the vector with a one in the $j$-th location and zeros elsewhere.
Then, abusing notation,
$\rho(v) e_j = v(x)x^j (\mod~ f)$
and $K_C(v)e_j = C^j v = x^jv(x) (\mod~ f)$.  Since this is true for
arbitrary $j$ the lemma is proved.
\end{proof}

Let $U$ and $V$ be $b \times d$ and $d \times b$ matrices over $\F$.  Let $u_i$ 
be the $i$-th row of $U$ and $v_j$ be $j$-th column of $V$.  The sequence 
$U \bar{C} V$  of $b \times b$ matrices can be viewed as a $b \times b$
matrix of sequences whose $(i,j)$ element is equal, by the discussion above,
to $u_i \rho(v_j)^T$.  This matrix can be mapped to the $b \times b$ matrix 
over $\K$ whose $(i,j)$ element is the product $u_i v_j = \rho(v_j) u_i$.  
This is the outer product $UV^T$, with $U$ and $V$ viewed as a column vector over $\K$
and a row vector over $\K$ respectively.  Hence it is a rank one matrix over 
$\K$ provided neither $U$ nor $V$ is zero.  Since any rank one matrix is 
an outer product, this mapping can be inverted. There is a one to 
one association of sequences $U \bar{C} V$ with rank one matrices over $\K$.

To show that this mapping relates rank to the probability that the block
projection $U \bar{A} V$ preserves the minimum polynomial of $A$, we must show 
that 
if $\sum_{k=1}^t U_k \bar{C_f} V_k = 0$ then the corresponding sum of $t$ rank one
matrices over $\K$ is the zero 
matrix and vice versa.  This will be shown using
the fact that the transpose $\rho(v)^T$ is similar to $\rho(v)$.  While it is well
known that a matrix is similar to its transpose, we provide a proof in the following
lemma which constructs the similarity transformation and shows that the same similarity transformation works independent of $v$.

\begin{lem}
\label{thm:transpose}
Given an irreducible monic polynomial $f \in \F_q[x]$ of degree $d$, there exists a symmetric nonsingular matrix $P$ such that $P^{-1} \rho(v) P = \rho(v)^T$, for all $v \in\F_q^d$.
\end{lem}
\begin{proof}
We begin with $C_f$.  Every matrix is similar to it's transpose by a symmetric transform \citep{taussky1959}.  Let $P$ be a similarity transform such that $P^{-1}C_fP = C_f^T$.
Then $P^{-1}\rho(v)P = \sum_{k=0}^{d-1} v_k P^{-1}C_f^kP = \sum_{k=0}^{d-1} v_k (C_f^k)^T
= \rho(v)^T$.
\end{proof}

It may be informative to 
have an explicit construction of such a transform $P$.
It can be done with Hankel structure (equality on antidiagonals). 
Let $\H_n(a_1,a_2\ldots,a_n,$ $a_{n+1},\ldots, a_{2n-1})$ denote the $n\times n$ Hankel matrix with first row
$(a_1,a_2,$ $\ldots,a_n)$ and last row $(a_n,a_{n+1},\ldots,a_{2n-1})$.  For example
$ \H_2(a,b,c) = 
\left( \begin{matrix}
		a & b \\
		b & c \\
	\end{matrix}
\right).$
Then define $P$ as $P = -f_0 \oplus \H_{d-1}(f_2, f_3, \ldots, f_{d-1}, 1, 0, \ldots, 0)$.
A straightforward computation verifies $C_fP = PC_f^T$.  

\begin{lem}
\label{thm:map}
Given an irreducible monic polynomial $f \in \F_q[x]$ and it's extension field $\K = {\F_q[x]/\langle f(x) \rangle}$,
there exists a one-to-one, onto mapping from the $b\times b$ projections of $\bar{C_f}$ to $\K^{b\times b}$ 
that preserves zero sums, i.e. 
$\sum U_i C_f V_i = 0$ iff $\phi(\sum U_i C_f V_i)  = \sum \phi(U_i C_f V_i) = 0$.
\end{lem}
\begin{proof}
The previous discussion shows that the mapping $U \bar{C_f} V \rightarrow UV^T$
from $b\times b$ projections of $\bar{C_f}$ onto rank one matrices over $\K$
is one-to-one.  Let $u_{k,i}$ and $v_{k,j}$ be the $i$-th row of $U_k$
and and the $j$-th column of $V_k$, respectively. 
Let P be a matrix, whose existence follows from lemma \ref{thm:transpose}, such that $P^{-1}\rho(v)P = \rho(v)^T$.
Assume $\sum_{k=1}^t U_k \bar{C_f} V_k = 0$.  Then using lemma \ref{thm:krillov_reg} and properties of $\rho$
\begin{eqnarray*}
\dsum_{k=1}^t u_{k,i}^T \bar{C_f} v_{k,j} = 0 &~\Rightarrow~ & \dsum_{k=1}^t u_{k,i}^T \rho(v_{k,j}) = 0 
~~\Rightarrow ~~ \dsum_{k=1}^t u_{k,i}^T PP^{-1}\rho(v_{k,j}) PP^{-1}= 0 \\
&\Rightarrow & \dsum_{k=1}^t u_{k,i}^T P \rho(v_{k,j})^T P^{-1}= 0 
~~\Rightarrow ~~ \dsum_{k=1}^t (u_{k,i}^T P) \rho(v_{k,j})^T = 0 \\
&\Rightarrow & \dsum_{k=1}^t \tilde{u}_{k,i} \cdot v_{k,j} = 0, \mbox{where }
\tilde{u}_{k,i} = (u_{k,i}^T P). \\
\end{eqnarray*}
Let $\tilde{U}_k$ be the vector whose $i$-th row is $\tilde{u}_{k,i}$ then the
corresponding sum of outer projects $\sum_{k=1}^t \tilde{U}_k V_k^T = 0$.
Because $P$ is invertible, the argument can be done in reverse, and for any
zero sum of rank one matrices over $\K$ we can construct the corresponding sum of projections
equal to zero.
\end{proof}

Thus the probability that $\dsum_{i=1}^t U \bar{C_f} V = 0$ is the probability that randomly selected $t$-term outer products over $\K$ 
sum to zero.  
The next lemma on rank one updates provides basic results leading to these probabilities.

\begin{lem}
\label{lem:increase_rank}
Let $r,s \geq 0$ be given and consider rank one updates to $A = I_r \oplus 0_s$.
For conformally blocked column vectors
$u = (u_1^T, u_2^T)^T, v = (v_1^T,v_2^T)^T \in \F^r\times \F^s$.
we have that \\
$\rank(A +u v^T) = r-1$ if and only if $u_1^T v_1 = -1$ and $u_2, v_2$ are both zero, and\\
$\rank(A +u v^T) = r+1$ if and only if $u_2,v_2$ are both nonzero.
\end{lem}
\begin{proof}
Without loss of generality (orthogonal change of basis) we may restrict attention to the case that $u_1 = \alpha e_r$ and $u_2 = \beta e_{r+1}$, where 
$e_i$ is the $i$-th unit vector,
$\alpha = 0$ if $u_1 = 0$ and $\alpha = 1$ otherwise, and 
similarly for $\beta$ vis a vis $u_2$.
Suppose that in this basis $v = (w_1, \ldots, w_r, z_{r+1}, \ldots, z_n)^T$.
Then $$(I_r\oplus 0) + uv^T = 
\left(\begin{array}{ccc|ccc} 
1 & 
  \ldots & 0 & 0 & \ldots & 0\\
\vdots & 
  \ddots & \vdots & \vdots & \ddots & \vdots\\
\alpha w_1 & 
  \ldots & 1+\alpha w_r& \alpha z_{r+1} & \ldots & \alpha z_n\\ \hline
\beta w_1 & 
  \ldots & \beta w_r& \beta z_{r+1} & \ldots & \beta z_n\\
\vdots & 
  \ddots & \vdots & \vdots & \ddots & \vdots\\
0 & \ldots & 0 & 0 & \ldots & 0\\
\end{array}\right).
$$

The rank of $I_r + u_1v_1^T$ is $r-1$ just in case $u_1^T v_1 = -1$
\citep{Meyer00}.  In our setting this condition is that $\alpha w_r = -1$.  
We see that, for a rank of $r-1$, we must have that 
$\alpha w_r = -1$ and $\beta, z$ both zero.
For rank $r+1$ it is clearly necessary that both of $\beta, z$ are nonzero.
It is also sufficient because for $z_i \neq 0$ the order $r+1$ minor 
$I_{r-1}\oplus \left(\begin{matrix} 
1+\alpha w_r& \alpha z_{i} \\
\beta w_r& \beta z_{i} \\
\end{matrix}\right)$ has determinant $\beta z_i \neq 0$.
These conditions translate into the statements of the lemma before the change of basis.
\end{proof}

\begin{cor}
\label{ranks}
Let $A \in \F_q^{n \times n}$ be of rank $r$, and let $u,v$ be uniformly random in $\F_q^n$. Then,
\begin{enumerate}
\item
the probability that $\rank(A+uv^T) = r-1$ is  $$D(r) = \frac{ q^{r-1} (q^r - 1) }{ q^{2n} },$$
\item
the probability that $\rank(A+uv^T) = r+1$ is  $$U(r) = \frac{(q^{n-r}-1)^2}{q^{2(n-r)}},$$
\item
the probability that $\rank(A+uv^T) = r$ is  $$N(r) = 1 - D(r) - U(r) \geq \frac{2q^n - 1}{q^{2n}},$$
with equality when $r = 0$.
\end{enumerate}
\end{cor}

\begin{proof}
There exist nonsingular $R,S$ such that $RAS = I_r \oplus 0$ and $R(A+uv^T)S = I_r \oplus 0 + (Ru)(S^Tv)^T$.  Since $Ru$ and $S^Tv$ are uniformly random when $u,v$ are, we may assume without loss of generality that $A = I_r \oplus 0$.

For part 1, by corollary \ref{ranks}, the
rank of $I_r \oplus 0 + uv^T$ is less than $r$ only if both $u,v$ are zero in their last $n-r$ rows and $u^Tv = -1$.
For $u,v \in \F_q^r$, $u^Tv = -1$ only when $u \neq 0$ and we have, for the first $i$ such that $u_i \neq 0$, that $v_i = u_i^{-1}\sum_{j\neq i} u_jv_j$.  Counting, there are $q^r - 1$ possible $u$ and then $q^{r-1}$ $v$'s satisfying the conditions.  The stated probability follows.

For part 2, by the preceding lemma, the rank is increased only 
if the last $n-r$ rows of $u$ and $v$ are both nonzero. The probability of this
is $\frac{(q^{n-r}-1)^2}{q^{2(n-r)}}$.

For the part 3 inequality,
if the sign is changed and 1 is added to both sides, the inequality becomes 
$ D(r) + U(r) \leq \left(\frac{q^n - 1}{q^{n}}\right)^2$.
 Note that $U(r) = \left(\frac{q^n - q^r}{q^{n}}\right)^2$ and $D(r) \leq \left(\frac{q^r - 1}{q^{n}}\right)^2$.  Let $a = \left(\frac{q^n - q^r}{q^{n}}\right)$ and $b = \left(\frac{q^r - 1}{q^{n}}\right)$.  Note that $a$ and $b$ are positive.  Thus, it is obvious that $a^2 + b^2 \leq (a+b)^2$.  That is,
\[
U(r) + D(r) \leq \left(\frac{q^n - q^r}{q^{n}}\right)^2 + \left(\frac{q^r - 1}{q^{n}}\right)^2 \leq \left(\frac{q^n - 1}{q^{n}}\right)^2.
\]
Therefore, $N(r) = 1 - D(r) - U(r) \geq \frac{2q^n - 1}{q^{2n}}$.
\end{proof}

\begin{defn}
\label{def:q}
For $u_i,v_i$ uniformly random in $\F_q^b$, and $A = \sum_{i=1}^t u_iv_i^T \in \F_q^{n \times n}$,  let $\prank[q,n,t](r)$ denote the probability that $rank(A) = r$.
\end{defn}

\begin{cor}\label{recurrence}
Let $A = \dsum_{i=1}^t u_i v_i^T$, for uniformly random $u_i,v_i \in \F_q^n$, 
and let $D(r), U(r)$, and $N(r)$ be defined as described in corollary \ref{ranks}.  Let $\prank[t](r) = \prank[q,n,t](r)$ (definition \ref{def:q}).  Then,
$\prank(r)$ satisfies the recurrence relation
\begin{eqnarray*}
\prank(r) &=& \begin{cases}
  0, & \text{if $r < 0$ or $r > min(t,n)$} \\
  1, & \text{if $r = 0$ and $t = 0$} \\
  \phi_{t-1}(r), & otherwise,
\end{cases}
\end{eqnarray*}
where $\phi_t(r) = \prank(r-1)U(r-1) + \prank(r)N(r) + \prank(r+1)D(r+1)$; and $U(r),N(r),D(r)$ are defined as they are in corollary \ref{ranks}.
\end{cor}

\begin{proof}
The general recurrence is evident from the fact that a rank one update can change the rank by at most one, and that $\prank[0](0) = 1$.  The rank of the sum of $t$ rank one matrices cannot be greater than either $t$ or $n$, nor less than zero.  
\end{proof}

These probabilities apply as well to the preimage of our mapping (block projections of direct sums of companion matrices), which leads to the next theorem.
\begin{thm}
\label{thm:probcompanion}
Let $f \in \F_q[x]$ be an irreducible polynomial of degree $d$, and let $A= \dirsum_{i=1}^s C_f \in \F_q^{n \times n}$.  Then, $$\pmp(A) = 1 - \prank[s](0) \geq 1 - \prank[1](0),$$ where $\prank[s](r) = \prank[q^d,b,s](r)$ (definition \ref{def:q}).
\end{thm}

\begin{proof}  By lemmas \ref{thm:zero_sum_prob} and
 \ref{thm:map}, the probability that a $b\times b$ projection of $A$ fails is precisely $\prank[s](0)$.

For the inequality, in all cases $\prank[s](1) \leq 1-\prank[s](0)$.  Therefore, 

\begin{eqnarray*}
\prank[s+1](0) &=& \prank[s](0) \frac{2q^{db}-1}{q^{2db}} + \prank[s](1) \frac{q^d-1}{q^{2db}} \\
&\leq & \prank[s](0) \frac{2q^{db}-1}{q^{2db}} + (1-\prank[s](0)) \frac{q^d-1}{q^{2db}} \\
&=& \prank[s](0) \frac{2q^{db}-q^d}{q^{2db}} + \frac{q^d-1}{q^{2db}}.
\end{eqnarray*}

Let $g(x) = x\frac{2q^{db}-q^d}{q^{2db}} + \frac{q^d-1}{q^{2db}}$.  Since $q,d,b$ are positive integers, $g(x)$ is linear with positive slope.  
Probability $\prank[s](0)$ has range [0,1] and we have $\prank[s+1](0) \leq g(\prank[s](0)) \leq g(1) = \frac{2q^{db}-1}{q^{2db}} = \prank[1](0)$.  Therefore, $\prank[1](0) \geq \prank[s](0)$, for all $s \geq 1$.
\end{proof}

Theorem \ref{thm:probcompanion} generalizes theorem \ref{thm:block}.  That is, $$\pmp(C_f) = 1-Q_{q^d,b,1}(0) = (1-1/q^{db})^2,$$ where $f \in \F_q[x]$ is an irreducible polynomial of degree $d$.  
Theorem \ref{thm:probcompanion} makes clear that $\pmp(\dirsum_{i=1}^s C_f)$ is minimized when there is a single block, $s=1$.

The following theorem summarizes the exact computation of the probability that
the minimal polynomial of a matrix is preserved under projection, in terms of the elementary divisor structure of the matrix.

\begin{thm}
\label{thm:prob1}
Let $A \in \F_q^{n \times n}$ be similar to
$J = \dirsum_i \dirsum_j J_{f_i^{e_{i,j}}},$
where the $f_i$ are distinct irreducibles of degree $d_i$, and the $e_{i,j}$ are positive exponents, nonincreasing with respect to $j$.  Let $s_i$ be the number of $e_{i,j}$ equal to $e_{i,1}$.  Then,
$$\pmp(A) = \pmp(J) = \prod_i \pmp\left( \dirsum_j J_{f_i^{e_{i,j}}} \right) = \prod_i \pmp\left( \dirsum_{k=1}^{s_i} C_{f_i} \right) = \prod_i (1 - Q_{q^{d_i}, b, s_i}(0)).$$
\end{thm}
\begin{proof}
By lemma \ref{lem:sim}, $\pmp(A) = \pmp(J)$. By theorem \ref{thm:minpolysum}, $\pmp(J) = \prod_i \pmp\left( \dirsum_j J_{f_i^{e_{i,j}}} \right)$.  
By lemma \ref{lem:companion}, $\pmp\left( \dirsum_j J_{f_i^{e_{i,j}}} \right) = \pmp\left( \dirsum_{k=1}^{s_i} C_{f_i} \right)$.  
Finally, by theorem \ref{thm:probcompanion}, $\pmp\left( \dirsum_{k=1}^{s_i} C_{f_i} \right) = 1-Q_{q^{d_i}, b, s_i}(0)$.  
Therefore, $\pmp(A) = \prod_i (1-Q_{q^{d_i},b,s_i}(0))$.
\end{proof}

\subsection{Examples}\label{sec:examples}

This section uses theorem \ref{thm:prob1} to compute $\pmp(A)$ for several example matrices, and compares the probability for matrices with related but not identical invariant factor lists.
\[A_1 = \left(
        \begin{array}{rr|r|r|r}
          0 & 1 & 0 & 0 & 0 \\
          1 & 4 & 0 & 0 & 0 \\ \hline
          0 & 0 & 3 & 0 & 0 \\ \hline
          0 & 0 & 0 & 3 & 0\\ \hline
          0 & 0 & 0 & 0 & 3
        \end{array}
      \right),
A_2 = \left(
        \begin{array}{rrrr|r}
          0 & 1 & 0 & 0 & 0 \\
          1 & 4 & 0 & 0 & 0 \\
          1 & 0 & 0 & 1 & 0 \\
          0 & 1 & 1 & 4 & 0 \\ \hline
          0 & 0 & 0 & 0 & 3
        \end{array}
      \right),
A_3 = \left(
        \begin{array}{rr|rr|r}
          0 & 1 & 0 & 0 & 0 \\
          1 & 4 & 0 & 0 & 0 \\ \hline
          0 & 0 & 0 & 1 & 0 \\
          0 & 0 & 1 & 4 & 0 \\ \hline
          0 & 0 & 0 & 0 & 3
        \end{array}
      \right),
\]\[
A_4 = \left(
        \begin{array}{rr|rr|r}
          0 & 1 & 0 & 0 & 0 \\
          1 & 4 & 0 & 0 & 0 \\ \hline
          0 & 0 & 3 & 0 & 0 \\
          0 & 0 & 1 & 3 & 0 \\ \hline
          0 & 0 & 0 & 0 & 3
        \end{array}
      \right),
A_5 = \left(
        \begin{array}{r|r|r|r|r}
          1 & 0 & 0 & 0 & 0 \\ \hline
          0 & 2 & 0 & 0 & 0 \\ \hline
          0 & 0 & 3 & 0 & 0 \\ \hline
          0 & 0 & 0 & 4 & 0 \\ \hline
          0 & 0 & 0 & 0 & 5
        \end{array}
      \right),
 \]
where $A_i \in \F_7^{5 \times 5}$.  Let $f(x)$ and $g(x)$ be the irreducible polynomials $(x^2 + 3x + 6)$ and $(x+4)$ in $\F_7[x]$.  Let $F(A)$ denote the list of invariant factors of $A$ ordered largest to smallest. Thus,

\begin{eqnarray*}
F(A_1) &=& \{f(x)g(x), g(x), g(x)\}, \\
F(A_2) &=& \{f(x)^2g(x)\}, \\
F(A_3) &=& \{f(x)g(x), f(x) \}, \\
F(A_4) &=& \{f(x)g(x)^2, g(x) \}, \\
F(A_5) &=& \{(x+2)(x+3)(x+4)(x+5)(x+6)\}.
\end{eqnarray*}

By theorem \ref{thm:prob1}, 

\begin{eqnarray*}
\pmp[7,b](A_1) &=& \pmp[7,b](C_f) \pmp[7,b](C_g\oplus C_g\oplus C_g) = (1-\prank[7^2,b,1](0))(1-\prank[7,b,3](0)), \\
\pmp[7,b](A_2) &=& \pmp[7,b](J_{f^2}) \pmp[7,b](C_g) = (1-\prank[7^2,b,1](0))(1-\prank[7,b,1](0)), \\
\pmp[7,b](A_3) &=& \pmp[7,b](C_f \oplus\ C_f) \pmp[7,b](C_g) = (1-\prank[7^2,b,2](0))(1-\prank[7,b,1](0)), \\
\pmp[7,b](A_4) &=& \pmp[7,b](C_f) \pmp[7,b](J_{g^2} \oplus C_g) = (1-\prank[7^2,b,1](0))(1-\prank[7,b,1](0)), \\
\pmp[7,b](A_5) &=& \prod_{i=1}^5 \pmp[7,b](C_{x-i+7}) = \prod_{i=1}^5 (1-\prank[7,b,1](0)).
\end{eqnarray*}

\begin{table}[H]
 \begin{center}
 \caption{\label{table:blocksize}$\pmp[7,b](A_i)$ vs $b$}
 \begin{tabular}{|c||c|c|c|c|}
   \hline
   & b=1 & b=2 & b=3 & b=4 \\
   \hline
   $\pmp[7,b](A_1)$ & 0.820 & 0.998 & 0.99998 & 0.9999996 \\
   $\pmp[7,b](A_2)$ & 0.705 & 0.959 & 0.994   & 0.9992 \\
   $\pmp[7,b](A_3)$ & 0.719 & 0.960 & 0.994   & 0.9992 \\
   $\pmp[7,b](A_4)$ & 0.705 & 0.959 & 0.994   & 0.9992 \\
   $\pmp[7,b](A_5)$ & 0.214 & 0.814 & 0.971   & 0.996 \\
   \hline
 \end{tabular}
 \end{center}
\end{table}

By part 3 of theorem \ref{thm:block}, $(1-\prank[7^2,b,1](0)) = (1-1/7^{2b})^2$ and $(1-\prank[7,b,1](0)) = (1-1/7^b)^2$.  Using
the recurrence relation in corollary \ref{recurrence}, 
we may compute $\prank[7,b,3](0)$ and $\prank[7^2,b,2](0)$.
Table \ref{table:blocksize} shows the resulting probabilities. 
Observe that $\pmp[7,b](A_i)$ increases as $b$ increases.

These five examples illustrate the effect of varying matrix structure and block size on $\pmp(A_i)$.  By theorem \ref{thm:probcompanion}, $\pmp[7,b](C_g \dirsum C_g \dirsum C_g) > \pmp[7,b](C_g)$ and $\pmp[7,b](C_f \dirsum C_f) > \pmp[7,b](C_f)$.  By theorem \ref{thm:prob1}, $\pmp[7,b](J_{f^2}) = \pmp[7,b](C_f)$ and $\pmp[7,b](J_{g^2} \dirsum C_g) = \pmp[7,b](C_g)$.  Therefore, $\pmp[7,b](A_1) > \pmp[7,b](A_2)$ and similarly $\pmp[7,b](A_3) > \pmp[7,b](A_2) = \pmp[7,b](A_4)$.  
Finally, since $(1-1/7^b)^2 < 1$ and $(1-1/7^b)^2 < (1-1/7^{2b})^2$, $\pmp[7,b](C_{h_1} \dirsum C_{h_2}) < \pmp[7,b](C_g)$ and $\pmp[7,b](C_h) < \pmp[7,b](C_f)$, for any linear $h_1(x),h_2(x),h(x) \in \F_7[x]$.  Therefore, $\pmp[7,b](A_5)$ has the minimal probability amongst the examples and in fact has the minimal
probability for any $5 \times 5$ matrix.  The worst case bound is explored
further in the following section.

\section{Probability Bounds: Matrix of Unknown Structure}\label{sec:worst}
Given the probabilities determined in section~\ref{sec:prob} 
of minimum polynomial preservation under projection, 
it is intuitively clear that the lowest probability of success would occur when there are many elementary divisors and the degrees of the irreducibles are as small as possible.  This is true and is precisely stated in theorem \ref{thm:worst} below.  First we need several lemmas concerning direct sums of Jordan blocks.

For $A \in \F_q^{n \times n}$, 
as before, 
$\pmp(A)$ denotes the 
probability that\\ $\minpoly(A) = \minpoly(U\bar AV)$, where $U \in \F_q^{b\times n}$ and $V \in \F_q^{n\times b}$ are uniformly random.

\begin{lem}
\label{lem:prob1}
Let $f$ be an irreducible polynomial over $\F_q$, let $e_1 = \ldots = e_s > e_{s+1} \geq \ldots \geq e_t$ be a sequence of exponents for $f$, 
and let $b$ be the projection block size.  Then 
	\[
	 \pmp(J_{f^{e_1+\cdots + e_t}}) \leq
	 \pmp(J_{f^{e_1}} \oplus \cdots \oplus J_{f^{e_t}}) 
	 = \pmp(J_{f^{e_1}} \oplus \cdots \oplus J_{f^{e_s}})
	\]
\end{lem}

\begin{proof}  
This follows from part 3 of 
theorem \ref{thm:block}, and theorems \ref{thm:probcompanion} and \ref{thm:prob1}, since\\ $\pmp(J_{f^{e_1+\cdots + e_t}}) = 1-\prank[1](0) \le 1-\prank[s](0) = \pmp(J_{f^{e_1}} \oplus \cdots \oplus J_{f^{e_t}})$.  
\end{proof}

\begin{lem}
\label{lem:prob2}
Let $f$ be an irreducible polynomial over $\F_q$ of degree $d$, let 
$f_1,\ldots,f_e$ be distinct irreducible polynomials of degree $d$ over 
$\F_q$,  
and let $b$ be the projection block size.  Then 
  \[
     \pmp(J_{f_1} \oplus \cdots \oplus J_{f_e}) \le \pmp(J_{f^e}).
  \]
\end{lem}
\begin{proof}  This follows from theorem
\ref{thm:minpolysum} and part 3 of theorem \ref{thm:block}, since $\pmp(J_{f_1} \oplus \cdots \oplus J_{f_e}) = \prod_{i=1}^e \pmp(J_{f_i})$ and
$\pmp(J_{f^e}) = \pmp(J_{f_i}) = (1-1/q^{db})^2 < 1$.
\end{proof}

\begin{lem}
\label{lem:prob3}
Let $f_1$ and $f_2$ be irreducible polynomials over $\F_q$ of degree $d_1$ 
and $d_2$ respectively 
and let $b$ be any projection block size.  Then, if $d_1 \le d_2$,
  \[
     \pmp(J_{f_1}) \le \pmp(J_{f_2}).
  \]
\end{lem}
\begin{proof}  The follows again from Part 3 of theorem
\ref{thm:block} since 
$(1-1/q^{d_1 b})^2 \le (1-1/q^{d_2 b})^2 $.
\end{proof}

Recall the definition:
$\pmpmin(n) = \min(\{ \pmp(A) | A \in \F_q^{n \times n}\}).$ 
This is the worst case probability that an $n\times n$ matrix has minimal polynomial preserved by uniformly random projection to a $b\times b$ sequence.     
In view of the above lemmata, for the lowest probability of success we
must look to 
matrices with the maximal number of elementary divisors.  
Define $L_q(m)$ to be the number of monic irreducible polynomials of degree $m$ in $\F_q[x]$. 
By the well known formula of \citet{Gauss81},
$$L_q(m) = 1/m \sum_{d~|~m} \mu(m/d)q^{d},$$
where $\mu$ is the M\"obus function.  Asymptotically $L_q(m)$ converges to $q^m/m$.  By definition, 
$\mu(a) = (-1)^k$ for square free $a$ with $k$ distinct prime factors and $\mu(a) = 0$ otherwise.
The degree of the product of all the monic irreducible polynomials of degree $d$ is then $dL_q(d)$.  
When we want to have a maximal number of irreducible factors in a product of degree $n$,
we will use $L_q(1), L_q(2),\ldots, L_q(m-1)$ etc., until the contribution of $L_q(m)$ no longer fits within the degree $n$.  
In that case we finish with as many of the degree $m$ irreducibles as will fit.  
For this purpose we adopt the notation $$L_{q}(n,m) := \min\left(L_q(m),\left\lfloor 
\frac{r}
{m}\right\rfloor\right), \for 
r = n - \sum_{d=1}^{m-1} d L_q(d)
.$$

\begin{thm}
\label{thm:worst}
Let $\F = \F_q$ be the field of cardinality $q$. 
For the $m$ such that $\sum_{d=1}^{m-1} d L_q(d) \le n < \sum_{d=1}^m d L_q(d)$,
$$\pmpmin(n) = \prod_{d=1}^m (1-1/q^{db})^{2L_{q}(n,m)}.$$  
Let 
$r = n - \sum_{d=1}^{m} d L_{q}(m,d)$. When $r \equiv 0 ~(\mod~ m)$, the minimum occurs for 
those matrices whose elementary divisors
are irreducible (not powers thereof), distinct, and with 
degree as small as possible.  When 
$ r \not\equiv 0~(\mod~ m)$ the minimum occurs when the elementary divisors involve exactly the same
irreducibles as in the $r \equiv 0~(\mod~m)$ case, but with some elementary divisors being powers so that that the total degree is brought to $n$.
\end{thm}

\begin{proof}
Let $A \in \F_q^{n \times n}$ and let 
$f_1^{e_1},\ldots,f_t^{e_t}$ be irreducible powers equal to the invariant 
factors of $A$.  If $\pmp(A)$ is minimal, then by 
lemmas \ref{lem:prob1},\ref{lem:prob2},\ref{lem:prob3} we can 
assume that the $f_i$ are distinct and have as small degrees as possible.  
Since $\sum_{d=1}^{m-1} d L_q(d) \le n < \sum_{d=1}^m d L_q(d)$, this 
assumption implies that all irreducibles of degree less than $m$ have been 
exhausted.  

If additional polynomials of degree $m$ can be added to obtain an 
$n \times n$ matrix, this will lead to the minimal probability since adding 
any irreducibles of higher degree will, by theorem \ref{thm:block}, reduce the total 
probability by a lesser amount.  In this case all of the exponents, $e_i$ 
will be equal to one.  If $r$ is not 0, then an 
$n \times n$ matrix can be obtained by increasing some of the exponents, 
$e_i$, without changing the probability.  This, again by theorem \ref{thm:block}, will 
lead to a smaller probability than those obtained by removing smaller degree 
polynomials and adding a polynomial of degree $m$ or higher.
\end{proof}

\subsection{Approximations}

Theorem \ref{thm:worst} can be simplified using the approximations  $L_q(m) \approx q^m/m$ and
$(1-1/a)^a \approx 1/e$.
\begin{cor}\label{cor:worst}
For field cardinality $q$, matrix dimension $n$, and projection block dimension $b$,
$$\pmpmin(n) \approx e^{-\frac{2}{q^b}H_m},$$
where $H_m$ is the $m$-th harmonic number.
$\Box$
\end{cor}
Also, 
for large primes, the formula of theorem \ref{thm:worst} simplifies quite a bit because there are plenty of small degree irreducibles.
In the next corollary we 
consider (a) the case in which there are $n$ linear irreducibles and (b) a situation in which the worst case probability will be defined by linear and quadratic irreducibles.
\begin{cor}\label{cor:simp}
For field cardinality $q$, matrix dimension $n$, and projection block dimension $b$,
if $q \geq n$ then $$\pmpmin(n) = (1-1/q^{b})^{2n} \approx e^{-2n/q^b}.$$  
If $n >q \geq n^{1/2}$ then 
$$\pmpmin(n) = (1-1/q^{b})^{2q}(1-1/q^{2b})^{n-q} \approx e^{-(2/q^{b-1} + (n - q)/q^{2b})}.$$  
\algend
\end{cor}

\subsection{Example Bound Calculations and Comparison to Previous Bounds}\label{examplebounds}

When $b = 1$ and we are only concerned with projection on one side, the first formula of corrolary \ref{cor:simp} simplifies to 
$(1-1/q)^n = (1 - n/q + \ldots)$. The bound given by Kaltofen and Pan \citep{KaPa91,Kaltofen:1991:SSLS} 
for the probability of $\minpoly(u\bar Av) = \minpoly(\bar Av)$ 
is the first two terms of this expansion, though developed with a very different proof.

For small primes,
\citet{Wiedemann86}(proposition 3) treats the case $b = 1$ and he fixes the projection on one side because he is interested in linear system solving and thus in the sequence $\bar Ab$, for fixed $b$.  For small $q$, his formula, $1/(6\log_q(n))$, computed with some approximation, is nonetheless quite close to our exact formula.  
However as $q$ approaches $n$ the discrepancy with our exact formula increases.
At the large/small crossover, $q = n$, Kaltofen/Pan's lower bound is 0, Wiedemann's is $1/6$, and ours is $1/e.$
The Kaltofen/Pan probability bound improves as $q$ grows larger from $n$.  The Wiedemann bound becomes more accurate as $q$ goes down from $n$.  
But the area $q \approx n$ is of some practical importance.
In integer matrix algorithms where the finite field used is a choice of the algorithm,
sometimes practical considerations of efficient field arithmetic encourages the use of primes in the vicinity of $n$.  For instance, exact arithmetic in double precision and using BLAS \citep{fflas08} works well with $q \in 10^6..10^7$.  Sparse matrices of order $n$ in that range are tractable.  Our bound may help justify the use of such primes.

\begin{figure}[H]
\centering
\includegraphics[width=0.8\textwidth,trim=1.4in 5.3in 1.4in 1in]{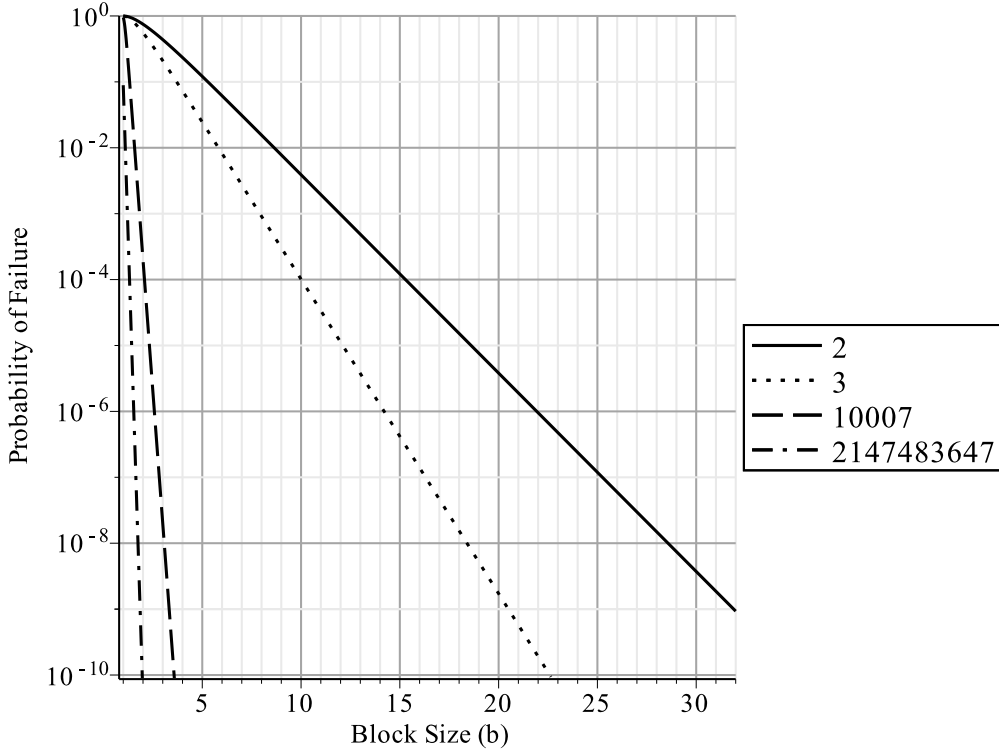}
\caption{Probability of Failure to Preserve Minimal Polynomial ($1-\pmpmin(10^8)$) 
vs Block Size and Field Cardinality}
\label{probability_plot}
\end{figure}

But the primary value we see in our analysis here is the understanding it gives of the value of blocking, $b > 1$.
Figure \ref{probability_plot} shows the bounds for the worst case probability that a random projection will preserve the minimal polynomial of a matrix $A \in \F_q^{10^8 \times 10^8}$ for various fields and projection block sizes.
It shows that the probability of finding the minimal polynomial correctly under projection converges rapidly to 1 as the projected block size increases.  

\section{Conclusion}\label{conclusion}
We have drawn a precise connection between the elementary divisors of a matrix and the probability that a random projection, as done in the (blocked or unblocked) Wiedemann algorithms, preserves the minimal polynomial.  
We provide sharp formulas both for the case where the elementary divisor structure of the matrix is known (theorem \ref{thm:minpolysum} and theorem \ref{thm:prob1}) and for the worst case (theorem \ref{thm:worst}).  
As indicated in figure \ref{probability_plot} for the worst case, a blocking size of 22 assures probability of success greater than $1 - 10^{-6}$ for all finite fields and all matrix dimensions up to $10^8$.  The probability decreases very slowly as matrix dimension grows and, in fact, further probability computations show
that the one in a million bound on failure applies to blocking size 22 with much larger matrix dimensions as well.
Looking forward, it
would be worthwhile to extend the analysis to apply to the determination of additional invariant factors.
Blocking is known to be useful for finding and exploiting them. For example, some rank and Frobenius form algorithms are based on block Wiedemann \citep{Eberly00b,Eberly00a}.  
Also, we have not addressed preconditioners.  
The preconditioners such as diagonal, Toeplitz, butterfly \citep{CEKTSV02}, either apply only for large fields or have only large field analyses.  
One can generally use an extension field to get the requisite cardinality, but the computational cost is high.  
Block algorithms hold much promise here and analysis to support them over small fields will be valuable.

\bibliographystyle{elsart-harv}

\end{document}